\documentclass{article}
\usepackage{fullpage}
\usepackage{enumitem}
\usepackage{mathrsfs} 
\usepackage{comment} 
\usepackage[hidelinks]{hyperref}

\usepackage{graphicx}
\usepackage{amsthm}
\usepackage{amsmath}
\usepackage{amsfonts}

\newtheorem{theorem}{Theorem}
\newtheorem{claim}{Claim}
\newtheorem{lemma}[theorem]{Lemma}
 
\theoremstyle{definition}

\newtheorem{observation}[theorem]{Observation}

\newcommand{\RR}{\mathbb R}
\newcommand{\KKn}{\mathcal K}
\newcommand{\KK}{\KKn^+}

\newcommand{\ee}{\varepsilon}

\newcommand{\dist}{\textrm{dist}}
\newcommand{\ball}[2]{B_{#2}(#1)}
\newcommand{\tool}{\mathcal D}

\DeclareMathOperator{\Arg}{Arg}
\DeclareMathOperator{\Int}{Int}
\DeclareMathOperator{\Ext}{Ext}

\DeclareMathOperator{\rad}{radius}

\begin{document}


\title{Disks in Curves of Bounded Convex Curvature}

\markright{Disks in Curves of Bounded Convex Curvature}
\interfootnotelinepenalty=10000000
\author{Anders Aamand\footnote{Basic Algorithms Research Copenhagen (BARC), University of Copenhagen. BARC is supported by the VILLUM Foundation grant 16582. Emails: \href{mailto:aa@di.ku.dk}{\texttt{aa@di.ku.dk}}, \href{mailto:miab@di.ku.dk}{\texttt{miab@di.ku.dk}}, \href{mailto:mikkel2thorup@gmail.com}{\texttt{mikkel2thorup@gmail.com}}.}, Mikkel Abrahamsen\footnotemark[1], and Mikkel Thorup\footnotemark[1]}

\date{August 30, 2019}

\maketitle

\begin{abstract}
We say that a simple, closed curve $\gamma$ in the plane has bounded convex curvature if for every point $x$ on $\gamma$, there is an open unit disk $U_x$ and $\ee_x>0$ such that $x\in\partial U_x$ and $\ball{x}{\ee_x}\cap U_x\subset\Int\gamma$.
We prove that the interior of every curve of bounded convex curvature contains an open unit disk.
\end{abstract}

\begin{figure}[h]
\centering
\includegraphics{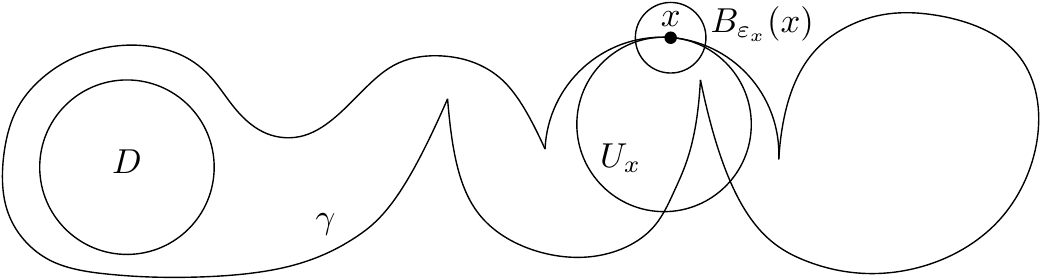}
\caption{A curve $\gamma$ of bounded convex curvature together with a unit disk $D$ in its interior, the existence of which is guaranteed by Theorem~\ref{MAINTHM}.}
\label{bccFig}
\end{figure}

\section{Introduction}

Consider a \emph{Jordan curve} $\gamma$, that is, a simple, closed curve in the plane. We will denote by $\Int \gamma$ and $\Ext \gamma$, respectively, the interior and exterior of $\gamma$.
We say that $\gamma$ has \emph{bounded convex curvature} if for every point $x$ on $\gamma$, there is an open unit disk $U_x$ and $\ee_x>0$ such that
\begin{align}
x\in\partial U_x\quad\text{and}\quad \ball{x}{\ee_x}\cap U_x\subset\Int\gamma. \label{bccCond} 
\end{align}
Here $\ball{x}{\ee}$ is the open disk with center $x$ and radius $\ee$. 
Similarly, we say that $\gamma$ has \emph{bounded concave curvature} if for every point $x$ on $\gamma$, there is an open unit disk $V_x$ and $\ee_x>0$ such that
\begin{align}
x\in\partial V_x\quad\text{and}\quad \ball{x}{\ee_x}\cap V_x\subset\Ext\gamma. \label{bccCond2}
\end{align}
Finally we say that a curve has \emph{bounded curvature} if it has both bounded convex and concave curvature. Curves of bounded convex curvature are the focus of this article.
When we say that $\gamma$ is a curve of bounded convex curvature it will always be understood that $\gamma$ is a Jordan curve.
Figure~\ref{bccFig} shows an example of a curve of bounded convex curvature.
Note that there may be points on a curve of bounded convex (or concave) curvature where the tangent to the curve is not defined.
Our main goal is to prove the following theorem (generalizing a theorem by Pestov and Ionin~\cite{pestov1959largest} that we shall discuss later):

\begin{theorem}\label{MAINTHM}
The interior of any curve of bounded convex curvature contains an open unit disk.
\end{theorem}

The  theorem does not hold if we replace the word ``convex'' with ``concave'' --- any circle of radius smaller than 1 provides a counterexample.

An appealing property of curves of bounded convex curvature is that they can be composed as described in the following observation (also see Figure~\ref{fig:union}).

\begin{figure}
\centering
\includegraphics{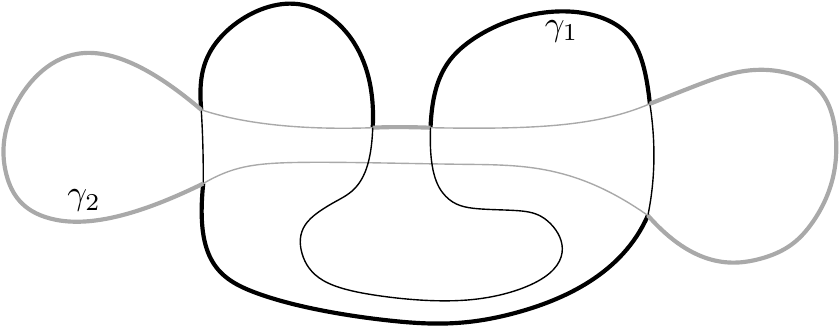}
\caption{Illustration for Observation~\ref{obs:union}.
The fat curve is the composition $\gamma_3$ of $\gamma_1$ (black) and $\gamma_2$ (gray).}
\label{fig:union}
\end{figure}

\begin{observation}\label{obs:union}
Let $\gamma_1$ and $\gamma_2$ be two curves of bounded convex curvature.
Consider the unbounded connected component $R$ of $\Ext\gamma_1\cap\Ext\gamma_2$.
If the boundary $\partial R$ is a Jordan curve $\gamma_3$, then $\gamma_3$ has bounded convex curvature.
\end{observation}

Note that this result does not hold for curves of bounded curvature.
Indeed the Jordan curves $\gamma_1$ and $\gamma_2$ in Figure~\ref{fig:union} both have bounded curvature, whereas their composition $\gamma_3$ only has bounded convex curvature.

In Section~\ref{sec:app} we will explain how curves of bounded convex curvature naturally arise in problems related to computer-aided manufacturing, but first we discuss related work.

\subsection{Related work}\label{relatedWork}

All previously studied notions of bounded curvature are more restrictive, and moreover defined in terms of a parameterization of the curve, contrary to our notion of bounded convex curvature.
The curvature is often defined for curves $\gamma$ that are two times continuously differentiable and parameterized by arclength.
Then the (unsigned) curvature at $s$ is simply $\|\gamma''(s)\|$, and a curve $\gamma$ is defined to have bounded curvature if $\|\gamma''(s)\|\leq 1$ for all $s$.
We say that such curves have \emph{strongly bounded curvature} in order to avoid confusion with the curves of bounded curvature introduced in this article.
Pestov and Ionin~\cite{pestov1959largest} proved that the interior of every curve of strongly bounded curvature contains an open unit disk.
We denote this theorem as the \emph{Pestov--Ionin theorem}.

The Pestov--Ionin theorem has often been applied to problems in robot motion planning and related fields~\cite{abrahamsen2016finding,agarwal2002curvature,ahn2012reachability,ayala2015length,lazard1998complexity}.
In Section~\ref{sec:app}, we describe how curves of bounded convex curvature naturally arise in problems related to pocket machining.

Dubins~\cite{dubins1957curves} introduced the class of curves of \emph{bounded average curvature} as the curves $\gamma$ parameterized by arclength that are differentiable such that for all $s_1,s_2$, we have
\begin{align}\label{bacCond}
\|\gamma'(s_1)-\gamma'(s_2)\|\leq |s_1-s_2|.
\end{align}
For a curve $\gamma$ of bounded average curvature, the second derivative $\gamma''$ is not necessarily defined everywhere, but since $\gamma'$ satisfies the Lipschitz condition~\eqref{bacCond}, it follows that $\gamma''$ is defined almost everywhere.
Dubins mentioned that if $\gamma$ is a curve parameterized by arclength for which $\gamma''$ exists everywhere, then $\gamma$ has bounded average curvature if and only
if $\gamma$ has strongly bounded curvature.
Ahn et al.~\cite{ahn2012reachability} proved that the Pestov--Ionin theorem holds for curves of bounded average curvature, and their proof is analogous to that of Pestov and Ionin.
In particular, both proofs rely on the curve $\gamma$ being rectifiable, i.e., having finite length.
However, it is not at all clear from our more general definition that a curve $\gamma$ of bounded convex curvature is rectifiable, so that approach cannot easily be applied in our case.
Instead, our proof shows that if $\Int \gamma$ contains no unit disk, then there exists an $\alpha>0$ such that $\Int \gamma$ contains infinitely many pairwise disjoint disks of radius $\alpha$. As $\gamma$ is bounded, this is of course a contradiction.

Pankrashkin~\cite{Pankrashkin2015} gave a proof that the interior of a smooth Jordan curve of strongly bounded curvature has area at least $\pi$.
This of course follows from the Pestov--Ionin theorem, but Pankrashkin proved it by other means.

Note that the requirement on the curvature of curves of strongly bounded and bounded average curvature is completely symmetric with respect to the curve turning to the left and  right when traversed in positive direction.
In contrast to that, Howard and Treibergs~\cite{howard1995reverse} introduced a class $\KK$ of curves satisfying an asymmetric condition on the curvature, namely the curves $\gamma$ parameterized by arclength such that $\gamma'$ is absolutely continuous and
$$\langle \gamma'(s+h)-\gamma'(s),\mathbf n(s)\rangle\leq h$$
for all $s$ and $0<h<\pi$, where $\langle\cdot, \cdot\rangle$ is the dot-product and $\mathbf n(s)=\gamma'(s)^\bot$ is the unit normal.
They proved the Pestov--Ionin theorem for the Jordan curves in $\KK$.
Abrahamsen and Thorup~\cite{abrahamsen2016finding} introduced a class of Jordan curves related to $\KK$, but where the curves may have sharp concave corners without a well-defined tangent.
They gave a proof of a version of the Pestov--Ionin theorem for that class of curves.


It can be shown that each of the classes of Jordan curves mentioned here are subsets of the curves of bounded convex curvature.
It is therefore natural to investigate whether the Pestov--Ionin theorem holds for all curves of bounded convex curvature, which is exactly the statement of Theorem~\ref{MAINTHM}.

\section{Application to Pocket Machining}\label{sec:app}

\begin{figure}
\centering
\includegraphics[width=0.6\textwidth]{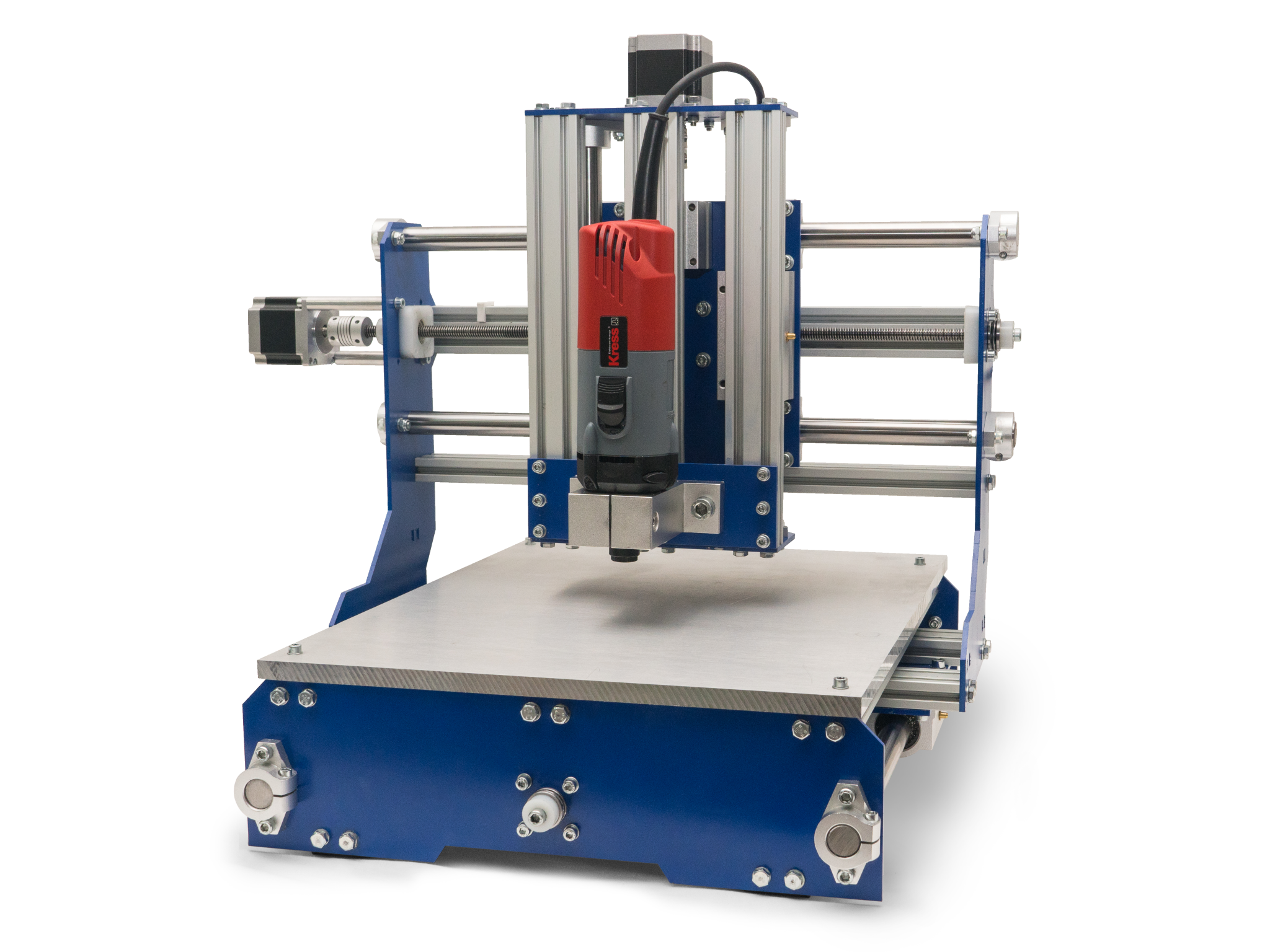}\quad
\includegraphics[width=0.35\textwidth]{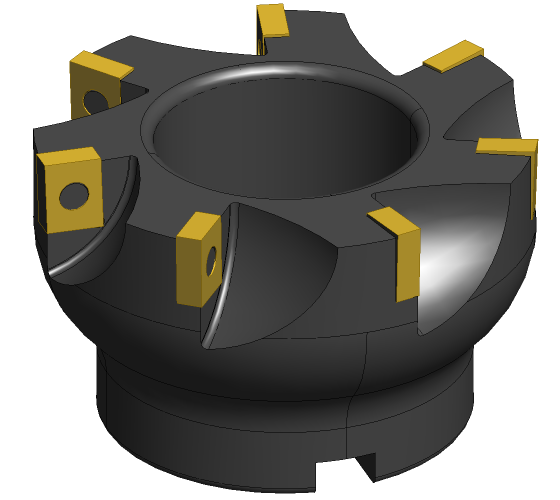}
\caption{Left: A milling machine.
The model is the Rabbit Mill v3.0 from SourceRabbit, who kindly provided permission to use the picture.
\copyright\ SourceRabbit.
Right: A milling tool. Picture by Rocketmagnet, licensed under CC BY-SA 3.0.}
\label{millFig}
\end{figure}

\begin{figure}
\centering
\includegraphics[width=\textwidth]{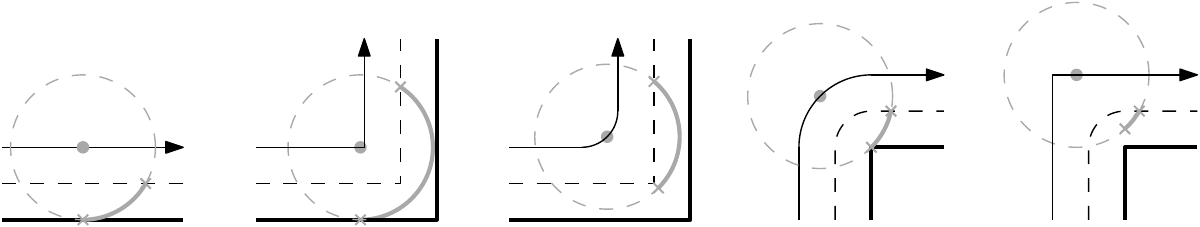}
\caption{In each of these four situations, the thick black curve is the boundary $\partial S$ of the pocket.
The remaining material in the pocket is ensured to be between the dashed black curve and $\partial S$.
The boundary of the tool $\tool$ is the dashed circle, and the solid part of the circle between the two crosses is the maximum part that can be in engagement with the material, i.e., the largest possible portion of the tool boundary cutting away material.
In the third picture, the convex corner on the path in the second picture has been rounded by an arc, thus bounding the convex curvature and reducing the maximum engagement.
The two rightmost pictures show two ways of going around a concave corner of $\partial S$.
In both cases, the maximum engagement is smaller than when the tool
follows a line segment of $\partial S$ (the case of the first picture).}
\label{cuttingCloseups}
\end{figure}

In this section we explain why it is sometimes natural to restrict oneself to curves of bounded convex curvature when choosing toolpaths for pocket machining. 
Pocket machining is the process of cutting out a pocket of some specified shape in a piece of material, such as a block of metal or wood, using a milling machine; see Figure~\ref{millFig} (left).

We are given a compact region $S$ of the plane whose boundary $\partial S$ is a Jordan curve.
The task is to remove the material in $S$ using a milling machine.
Suppose that we have already removed all material in $S$ except for a thin layer close to the boundary $\partial S$ (another coarser tool has removed most of the material, but is not fine enough to do the boundary itself).
In order to remove the remaining material, we are using a tool $\tool$, which can be thought of as a disk of some radius $r$, and we have to specify the toolpath.
The toolpath is a curve that the center of $\tool$ should follow, and the material removed is the area swept over by $\tool$ as it does so.
In practice, the tool has sharp teeth that cuts away the material as the tool spins at high speed; see Figure~\ref{millFig} (right).
The maximum thickness of the layer of remaining material is some fraction of the tool radius $r$, carefully chosen in order to limit the load on the tool.

It is an advantage if the tool center moves with constant speed while the tool is removing material, since that gives a higher surface quality of the resulting part.
Since the tool moves at constant speed, the load on the tool is heavier in a neighborhood around a convex turn than when it follows as straight line, since it has to remove more material per time unit.
In contrast to this, the load is lighter in a neighborhood around a concave turn.
See Figure~\ref{cuttingCloseups} for an illustration of this.
If the load is too heavy, the accuracy and surface quality will be inferior, and the tool can even break~\cite{han2015precise}.
It has been recommended to round the convex corners of the toolpath by circular arcs of a certain radius in order to decrease the load~\cite{choy2003corner,pateloup2004corner}.
In our terminology, this is the same as requiring the toolpath to have bounded convex curvature.
By restricting the convex curvature, we will inevitably leave more material that cannot be removed by the tool.
This can be removed by other tools that are more expensive to use in terms of machining time.

If the toolpath consists of all points at distance $r$ to $\partial S$, the concave curvature will be bounded by $1/r$, since the tool center will be ``rolling'' around any concave corner $v$ of $\partial S$ using a circular arc $A$ of radius $r$ (as in the fourth picture in Figure~\ref{cuttingCloseups}).
However, a recommended alternative way to get around $v$ is to follow the tangents to the endpoints of $A$ (as in the fifth picture in Figure~\ref{cuttingCloseups}---note that the tool will not remove any material when the center is in a neighborhood around the intersection point of the tangents).
Experience shows that this results in the corner $v$ being cut much sharper and more precisely~\cite{park2003mitered}.
This shows that the toolpaths arising in this context are required to have bounded convex curvature, whereas no bound can be given on the concave curvature.

Abrahamsen and Thorup~\cite{abrahamsen2016finding} studied the computational problem of computing the maximum region with a boundary of bounded convex curvature inside a given region in the plane, which defines the maximum region that can be cleared by the tool using a toolpath of bounded convex curvature.
A version of the Pestov--Ionin theorem (mentioned in the introduction) was used to establish the maximality of the region returned by the algorithm described in the article.

\section{Proving Theorem~\ref{MAINTHM}}
The proof of Theorem~\ref{MAINTHM} is by contradiction.
We assume that $\gamma$ is a curve of bounded convex curvature with an interior containing no open unit disk.
We then show that there exists an $\alpha>0$ such that $\Int \gamma$ contains infinitely many pairwise disjoint disks of radius $\alpha$. 
As $\Int \gamma$ is bounded, this is a contradiction.

To construct these disks we need to prove a special property of curves of bounded convex curvature, namely that the radii of disks $D\subset \Int \gamma$ having $|\gamma\cap\partial D|\geq 2$ are lower bounded by some constant $\eta>0$ depending only on $\gamma$.

Our first step is to set up an alternative condition that guarantees that $\gamma$ \emph{does not} have bounded convex curvature, as stated in Lemma~\ref{diamLemma} below.
We start with the following lemma; see Figure~\ref{fig:suffNotCond}.


\begin{figure}
    \centering
    \includegraphics{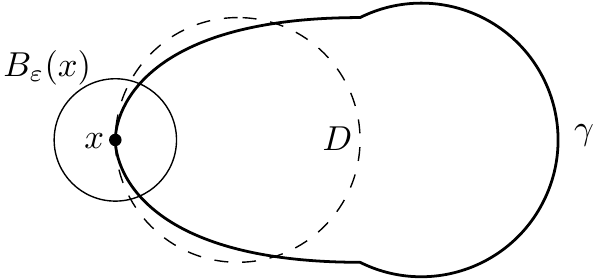}
    \caption{The situation described in Lemma~\ref{suffNotCond}. The curve $\gamma$ does not have bounded convex curvature.}
    \label{fig:suffNotCond}
\end{figure}

\begin{lemma}\label{suffNotCond}
Let $\gamma$ be a Jordan curve and consider a point $x$ on $\gamma$.
If there exists an open unit disk $D$ where $x\in\partial D$ such that
\begin{enumerate}
 \item
  there exists $ \ee>0$ such that $\ball{x}{\ee}\cap\Int\gamma\subset D$, and
\label{suffNotCond1}

\item
for all $\eta>0$ we have $\gamma\cap\ball{x}{\eta}\cap D\neq\emptyset$,
\label{suffNotCond2}
\end{enumerate}
then $\gamma$ does not have bounded convex curvature.
\end{lemma}

\begin{proof}
Assume for contradiction that $\gamma$ has bounded convex curvature, and choose $\ee_x>0$ and $U_x$ such that condition~\eqref{bccCond} in the definition of bounded convex curvature is satisfied for $x$.
We must show that for any unit disk $D$ with $x \in \partial D$, either condition~\ref{suffNotCond1} or~\ref{suffNotCond2} of Lemma~\ref{suffNotCond} fails.
Let $D$ be such a unit disk and suppose $\ee>0$ is such that condition~\ref{suffNotCond1} of the lemma is satisfied.
Let $\eta=\min\{\ee_x,\ee\}$.
Then,
$$\ball{x}{\eta}\cap U_x\subset
\ball{x}{\eta}\cap\Int\gamma \subset D.$$
This implies that $U_x=D$: Indeed, $U_x$ and $D$ are two unit disks with $x$ on the boundary, so if $U_x\neq D$, then $U_x\setminus D$ would contain points arbitrarily close to $x$.
Now $\ball{x}{\eta}\cap U_x \subset \Int \gamma$, so $\gamma \cap \ball{x}{\eta}\cap U_x  =\emptyset$. As $U_x=D$, condition~\ref{suffNotCond2} is not satisfied. This completes the proof.
\end{proof}


For any Jordan curve $\gamma$ and two distinct points $a$ and $b$ on $\gamma$, we denote by $\gamma[a,b]$ the closed interval on $\gamma$ from $a$ to $b$ in the positive direction. 
We may for example apply this notation to the boundary curve $\partial D$ for a disk $D$. 
By $\gamma(a,b)$, we denote the open interval $\gamma[a,b]\setminus\{a,b\}$.
While it might be intuitively clear what it means to traverse $\gamma$ in the positive or negative direction, we give a precise definition in Appendix~\ref{sec:appen}.

We require the following lemma which phrased informally states that if $\gamma$ is traversed positively, the interior of $\gamma$ is ``to the left'' of the curve.
\begin{lemma}\label{leftrightlemma}
Let $p$ be a point on a Jordan curve $\gamma$, and let $U$ be an open disk with center $p$, sufficiently small so that $\gamma$ is not contained in $U$.
The intersection of $U$ and $\gamma$ is a collection of intervals of $\gamma$ of which one, say $\gamma(a,b)$, contains $p$. Consider the Jordan curves
$$
\alpha^+=\gamma[a,b]\cup \partial U[b,a] \quad \text{and} \quad \alpha^-=\gamma[a,b]\cup \partial U[a,b]
$$
Then $\Int \gamma$ and $\Int \alpha^+$ coincide near $p$, that is, there exists a small disk $V\subset U$ centered at $p$ such that $\Int \gamma \cap V=\Int \alpha^+ \cap V$. Similarly $\Ext \gamma$ and $\Int \alpha^-$ coincide near $p$.
\end{lemma}
We believe the result to be standard but we were unable to find an equivalent one in the literature, phrased for completely arbitrary Jordan curves, e.g., with no assumptions on the differentiability of the curve.
We will provide a proof in Appendix~\ref{sec:appen}.

Suppose $\gamma$ is a Jordan curve, $a,b$ are distinct points on $\gamma$, and $D$ is an open disk satisfying $a,b\in\partial D$ and $\gamma[a,b]\cap\overline D=\{a,b\}$.
If $R$ is the open region bounded by the Jordan curve $\gamma[a,b]\cup \partial D[a,b]$, then either $D\subset R$ or $D\cap R=\emptyset$.
In the former case we say that $\gamma$ \emph{winds negatively} around $D$ from $a$ to $b$ and in the latter that $\gamma$ \emph{winds positively} around $D$ from $a$ to $b$.

\begin{figure}
\centering
\includegraphics[scale=1]{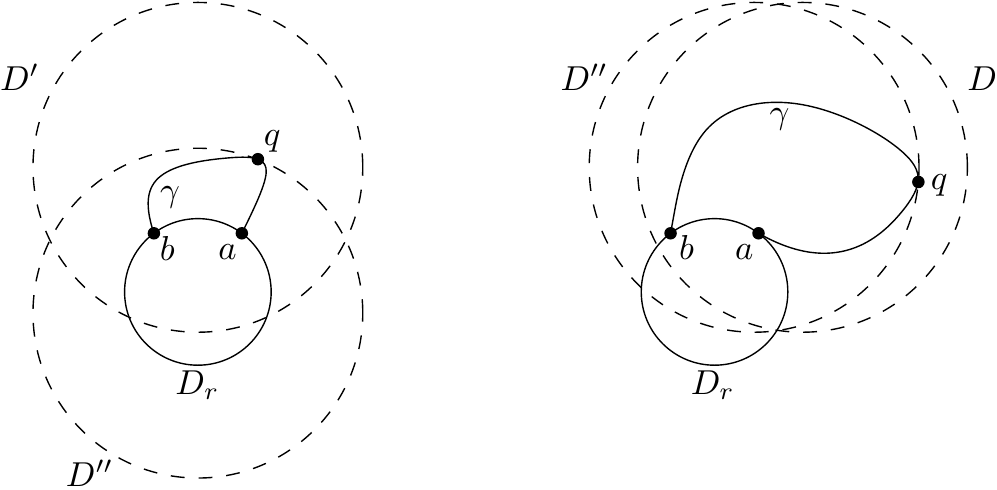}
\caption{The two cases in the proof of Lemma~\ref{diamLemma}.}
\label{fig:diamLemma}
\end{figure}

\begin{lemma}\label{diamLemma}
Let $\gamma$ be a Jordan curve and consider an interval $\gamma[a,b]$ of $\gamma$ such that $\gamma[a,b]$ is contained in an open unit disk $D$.
Suppose there is an open disk $D_r$ of radius $r\leq 1$ such that $\gamma[a,b]\cap\overline {D_r}=\{a,b\}$ and $\gamma$ winds positively around $D_r$ from $a$ to $b$.
Then $\gamma$ does not have bounded convex curvature.
\end{lemma}

\begin{proof}
The general outline of the proof is as follows: We first make a translation of $D$ into a disk $D''$ such that $\overline{D''}$ still contains $\gamma[a,b]$ and such that $\partial D''$ meets $\gamma(a,b)$ in at least one point. We then argue that we may choose a point $q\in \partial D'' \cap \gamma(a,b)$ for which Lemma~\ref{suffNotCond} applies to show that $\gamma$ does not have bounded convex curvature.

By translating and rotating we may assume about the coordinates that $D_r$ is centered at the origin and that $a=(s_0,t_0)$ and $b=(-s_0,t_0)$ for some $s_0,t_0$ with $s_0>0$ and $s_0^2+t_0^2=r^2$ (note that $t_0$ may be negative).
Suppose that $D$ is centered at $(s_1,t_1)$ and that $s_1\geq 0$ (the case $s_1\leq 0$ is dealt with in a symmetric way).
Also define the Jordan curve $\gamma'=\gamma[a,b]\cup \partial D_r[a,b]$.

We start the proof by showing the following two claims.
\begin{claim}\label{claim1}
Let $p$ be a point on the arc $\partial D_r(a,b)$. Let $m$ be the midpoint of segment $ab$ and $v=p-m$. Consider the ray $\ell_p=\{p+\alpha v: \alpha>0\}$. Then $\ell_p$ intersects $\gamma(a,b)$.
\end{claim}
\begin{proof}[Proof of Claim~\ref{claim1}]
Let $V$ be an open disk centered at $p$, so small that $V\cap \gamma[a,b]=\emptyset$.
Further let $c,d\in \partial D_r$ be such that $V\cap \partial D_r=\partial D_r(c,d)$. 
Then $V \backslash \partial D_r(c,d)$ is the disjoint union of two open connected sets $V_1$ and $V_2$ satisfying $V_1\subset D_r$ and $V_2\cap \ell_p \neq \emptyset$. 
Moreover, $V_1$ and $V_2$ are both subsets of $\RR^2\setminus \gamma'$ and, being connected, they are each fully contained in either $\Int \gamma'$ or $\Ext \gamma'$.
Now as $p\in \gamma'$ and $\gamma'=\partial (\Int \gamma')$ by the Jordan curve theorem, it follows that either $V_1 \subset \Int \gamma'$ or $V_2\subset \Int \gamma'$.
But by the assumption on the winding direction of $\gamma$ from $a$ to $b$, we have $V_1 \cap \Int \gamma' \subset D_r \cap \Int \gamma'=\emptyset$, and so $V_2\subset \Int \gamma'$.
It follows that $\ell_p \cap \Int \gamma' \neq \emptyset$.
Furthermore, we trivially have that $\ell_p \cap \Ext \gamma'\neq \emptyset$ and so $\ell_p$ must intersect $\gamma'$.
This cannot happen at a point of $D_r[a,b]$ so $\ell_p$ must intersect $\gamma(a,b)$ as claimed.
\end{proof}
\begin{claim}\label{claim2}
Let $D_0$ be an open disk satisfying that $\gamma[a,b]\subset \overline{D_0}$. Then $\gamma'\subset \overline{D_0}$.
\end{claim}
\begin{proof}[Proof of Claim~\ref{claim2}]
It clearly suffices to show that $\overline{D_0}$ contains $\partial D_r(a,b)$. Take any point $p\in \partial D_r(a,b)$ and consider the line $\ell_p=\{p+\alpha v: \alpha>0\}$ from Claim~\ref{claim1} that intersects $\gamma(a,b)$ in some point $p+\alpha_0 v$ where $\alpha_0>0$. 
Now by assumption $\overline{D_0}$ contains $\gamma[a,b]$, hence also $p+\alpha_0v$. Since $a,b\in \overline{D_0}$, and $\overline{D_0}$ is convex, $\overline{D_0}$ contains the midpoint $m$ of segment $ab$. Finally $p$ is on the line segment between $m$ and $p+\alpha_0v$ so by convexity $\overline{D_0}$ contains $p$. Since $p$ was arbitrary, this establishes the claim.
\end{proof}

We now let $D'$ be the disk $\ball{(0,t_1)}{1}$. We split the proof into two cases depicted in Figure~\ref{fig:diamLemma}.

\textbf{Case 1: $\gamma[a,b]\subset\overline{D'}$.}
In this case, we let $t''\in \RR$ be minimal such that the closure of the unit disk $D''=\ball{(0,t'')}{1}$ contains $\gamma[a,b]$.

Consider the set of intersection points $P=\gamma[a,b]\cap\partial D''$, which is nonempty by construction.
We claim that $P$ contains neither $a$ nor $b$. To see this, note that the ray $\ell=\{(0,t):t> r\}$ intersects $\gamma(a,b)$ by Claim~\ref{claim1}. Now if $\partial D''$ contained $a$ (and thus by symmetry $b$) then, as $r\leq 1$, we would have $\ell\cap \overline{D''}=\emptyset$ and hence that $\ell\cap \gamma(a,b)=\emptyset$, a contradiction.
We conclude that $P$ contains neither $a$ nor $b$.

The set $\gamma(a,b)\setminus P$ is nonempty as $a,b\notin\partial D''$, and consists of pairwise disjoint open arcs.
Let $q\in P$ be an endpoint of such an arc.
We will now show that $q$ and $D''$ satisfy the conditions of Lemma~\ref{suffNotCond}, from which it follows that $\gamma$ does not have bounded convex curvature.

That $\gamma \cap \ball{q}{\ee} \cap D''\neq \emptyset$ for all $\ee>0$ is immediate as $q$ is  an endpoint of one of the open arcs in $\gamma(a,b) \backslash P$.
It thus suffices to check condition~\ref{suffNotCond1}, as follows.
First note that either $\partial D_r[a,b]=\gamma'[a,b]$ or $\partial D_r[a,b]=\gamma'[b,a]$.
However, if $\partial D_r[a,b]=\gamma'[a,b]$, we could apply Lemma~\ref{leftrightlemma} to $\gamma'$ with $p=(0,r)$ to conclude that $D_r\subset \Int \gamma'$.
But we assumed that $D_r\cap \Int \gamma'=\emptyset$ and so it follows that $\partial D_r[a,b]=\gamma'[b,a]$.
Since $\gamma'=\partial D_r[a,b]\cup \gamma[a,b]=\gamma'[a,b]\cup \gamma'[b,a]$, we conclude that $\gamma[a,b]=\gamma'[a,b]$.
Another application of Lemma~\ref{leftrightlemma}, this time with $p=q$, gives that $\Int \gamma$ and $\Int \gamma'$ coincide locally near $q$, that is, there exists an $\ee>0$ such that $\Int \gamma \cap \ball{q}{\ee}=\Int \gamma' \cap \ball{q}{\ee}$.
Now $\overline{D''}$ contains $\gamma[a,b]$ and hence $\gamma'$ by Claim~\ref{claim2}. Thus, $\Int \gamma' \subset D''$ and it follows that $\Int \gamma \cap \ball{q}{\ee}\subset D''$, as desired.


\textbf{Case 2: $\gamma[a,b]\not\subset\overline{D'}$.}
In this case, let $s''>0$ be minimal such that the closure of $D''=\ball{(s'',t_1)}{1}$ contains $\gamma[a,b]$.
As $s''>0$, $\partial D''$ contains neither $a$ nor $b$.
Letting $P=\gamma[a,b]\cap\partial D''$, the same argument as in Case~1 finishes the proof.
\end{proof}
We are slowly setting up the stage for the proof of Theorem~\ref{MAINTHM}. Intuitively, the following lemma is unsurprising. The lemma will be helpful for checking one of the conditions of Lemma~\ref{diamLemma}, hence making it easier to apply.
\begin{lemma}\label{positivelemma}
Let $\gamma$ be a Jordan curve and $D$ an open disk contained in $\Int \gamma$. Suppose that $a,b$ are distinct points on $\gamma$ such that $\gamma[a,b]\cap \overline{D}=\{a,b\}$. Then $\gamma$ winds positively around $D$ from $a$ to $b$, that is, $D\subset\Ext(\gamma[a,b]\cup \partial D[a,b]$). Similarly, $D\subset \Int(\gamma[a,b]\cup \partial D[b,a])$.
\end{lemma}
\begin{proof}
We only prove the first statement in the theorem as the proof of the second part is similar.
Letting $\gamma'=\gamma[a,b]\cup \partial D[a,b]$ we must show that $D\subset \Ext \gamma'$. As $D\subset \Int \gamma$ it must hold that $\Int \gamma' \subset \Int \gamma$. Now either $\gamma'[b,a]=\gamma[a,b]$ or $\gamma'[b,a]=\partial D[a,b]$. Suppose first that $\gamma'[b,a]=\gamma[a,b]$ and let $p$ be any point on $\gamma(a,b)$. Applying Lemma~\ref{leftrightlemma} we find that $\Int \gamma'$ and $\Ext \gamma$ coincide near $p$. This is a contradiction as $\Int \gamma'\subset \Int \gamma$. It follows that $\gamma'[b,a]=\partial D[a,b]$. Now choose any point $p\in \partial D(a,b)$. Again applying Lemma~\ref{leftrightlemma} we find that $\Ext \gamma'$ and $\Int \partial D$ coincide near $p$. As $D\subset \RR^2\setminus \gamma'$, it immediately follows that $D\subset \Ext \gamma'$, as desired.
\end{proof}
Now we can prove that if $\gamma$ has bounded convex curvature, then certain maximal disks contained in $\Int \gamma$ cannot be too small.
\begin{lemma}\label{lowerLemma}
Let $\gamma$ be a curve of bounded convex curvature. There exists a constant $\eta>0$ with the following property: If $\ball{x}{r}\subset \Int \gamma$ is an open disk of radius $r$, and $\partial \ball{x}{r}$ meets $\gamma$ in at least two points, then $r\geq \eta$.
\end{lemma}
\begin{proof}
We show the contrapositive.
Suppose that no such $\eta$ exists and take a sequence of balls $\ball{x_n}{r_n}\subset\Int\gamma$ satisfying $|\gamma\cap\partial\ball{x_n}{r_n}|\geq 2$ for all $n$ and $\lim_{n\longrightarrow \infty }r_n=0$.
Further suppose that $r_n<1$ for all $n$.

For each $n$, let $a_n,b_n $ be two distinct points in $\gamma\cap\partial\ball{x_n}{r_n}$.
Since $\gamma \times \gamma$ is compact, we may assume that $(a_n,b_n)\longrightarrow (a,b)$ for some $(a,b)\in \gamma \times \gamma$ by passing to an appropriate subsequence. As $r_n \longrightarrow 0$ we must have that $a=b$.

Let $V$ be an open ball centered at $a$ of radius $1/2$. Then $V\cap \gamma$ is a collection of open intervals one of which, say $\gamma(c,d)$, contains $a$. Let $W\subset V$ be an open ball centered at $a$ and so small that $W\cap \gamma[d,c]=\emptyset$.

As $a_n,b_n \longrightarrow a$ we must have that $a_n,b_n\in W$ for $n$ sufficiently large. But then $a_n,b_n\in \gamma(c,d)$, so either $\gamma[a_n,b_n]\subset \gamma(c,d)$ or $\gamma[b_n,a_n]\subset\gamma(c,d)$. In particular, either  $\gamma[a_n,b_n]$ or $\gamma[b_n,a_n]$ is contained in an open unit disk. 

We now wish to apply Lemma~\ref{diamLemma} to show that this implies that $\gamma$ does not have bounded convex curvature. Assume without loss of generality that $n$ is such that $\gamma[a_n,b_n]$ is contained in an open unit disk. Now $\gamma[a_n,b_n]\setminus\partial \ball{x_n}{r_n}$ is a collection of open intervals of $\gamma$. Moreover, the collection is nonempty as otherwise $\gamma[a_n,b_n]\subset \partial \ball{x_n}{r_n}$, and as $r_n<1$ and $\ball{x_n}{r_n}\subset \Int \gamma$, this would violate the bounded convex curvature condition. We may thus choose distinct $a_n',b_n'$ such that $\gamma(a_n',b_n')$ is such an interval.

Since $\gamma$ winds positively around $\ball{x_n}{r_n}$ from $a_n'$ to $b_n'$ by Lemma~\ref{positivelemma}, we are in a position to apply Lemma~\ref{diamLemma} and we conclude that $\gamma$ does not have bounded convex curvature.




\end{proof}

For the proof of Theorem~\ref{MAINTHM} we will also need the following easy lemma.
\begin{lemma}\label{upperLemma}
Let $\gamma$ be a Jordan curve.
Let $r_0$ be the supremum over all $r>0$ such that $\Int \gamma$ contains an open disk of radius $r$. 
Then $\Int \gamma$ contains an open disk of radius $r_0$.
\end{lemma}
\begin{proof}
The proof is a standard compactness argument, using only that $\Int \gamma$ is a bounded open set. To be precise, let $f: \overline{\Int \gamma}  \longrightarrow \mathbb{R}_{\geq 0}$ be defined by 
\begin{align*}
f(x)=\sup \{r\geq 0: \ball{x}{r}\subset \Int \gamma\}.
\end{align*}
If we put $r':=f(x)$, then clearly $\ball{x}{r'}\subset \Int \gamma$, so we may in fact write $f(x)=\max \{r\geq 0: \ball{x}{r}\subset \Int \gamma\}$. 
Now, $|f(x)-f(y)|\leq \|x-y\|$ for any $x,y\in \overline{\Int \gamma}$ and thus $f$ is continuous.
Furthermore, $\sup \{f(x): x\in \overline{\Int \gamma}\}=r_0$, and since $\overline{\Int \gamma}$ is compact, $f$ attains this maximum at some point $x_0$.
But then $\ball{x_0}{r_0}\subset \Int \gamma$.
\end{proof}

We are now ready to prove Theorem~\ref{MAINTHM}.
\begin{proof}[Proof of Theorem~\ref{MAINTHM}]
Let $\gamma$ be a curve of bounded convex curvature and assume for contradiction that $\Int \gamma$ contains no open unit disk.

By Lemma~\ref{lowerLemma} and Lemma~\ref{upperLemma}, we may choose $\eta_1$ and $\eta_2$ with $0<\eta_1\leq 1-\eta_2<1$, such that any disk $D\subset \Int \gamma$ with $|\gamma\cap\partial D|\geq 2$ satisfies $\eta_1\leq \rad D\leq 1-\eta_2$.

Let $z$ be any point of $\gamma$ and let the disk $D_0\subset \Int \gamma$ be tangent to $U_{z}$ in $z$ and of maximal radius.
We note that $\gamma\cap\partial D_0$, apart from $z$, contains at least one other point. Otherwise, $\dist(\gamma\setminus B_{\ee_{z}}(z), D_0)>0$ and then we can enlarge $D_0$, contradicting the maximality of $D_0$.
Thus $\eta_1\leq \rad D_0\leq 1-\eta_2$.
The set $\gamma\setminus \partial D_{0}$ consists of some (at least two) open intervals of $\gamma$.
Let $x_0,y_0$ be distinct points on $\gamma$ such that $\gamma(x_0,y_0)$ is such an open interval.

In general for $n\geq 0$, we will recursively define distinct points $x_n,y_n\in \gamma$, and an open disk $D_n\subset \Int \gamma$ such that $\gamma[x_n,y_n]\cap \partial D_{n}=\{x_n,y_n\}$. Letting $A_n$ be the open region bounded by the Jordan curve $\gamma[x_n,y_n]\cup \partial D_n[x_n,y_n]$, the construction satisfies, for all $n\geq 0$, that
\begin{enumerate}[label=(\roman*)]
\item  \label{step1} $A_{n+1}\subset A_n$, and
\item \label{step2}  $A_n \setminus A_{n+1}$ contains an open disk $E_{n+1}$ of radius at least $\eta:=\min(\eta_1,\eta_2/2)$.
\end{enumerate}
The disks $(E_n)_{n>0}$ are pairwise disjoint and all contained in $\Int \gamma$, and moreover they have radius at least $\eta>0$. As $\Int \gamma$ is bounded, this gives the desired contradiction, thus completing the proof of the theorem. 


\begin{figure}
\centering
\includegraphics{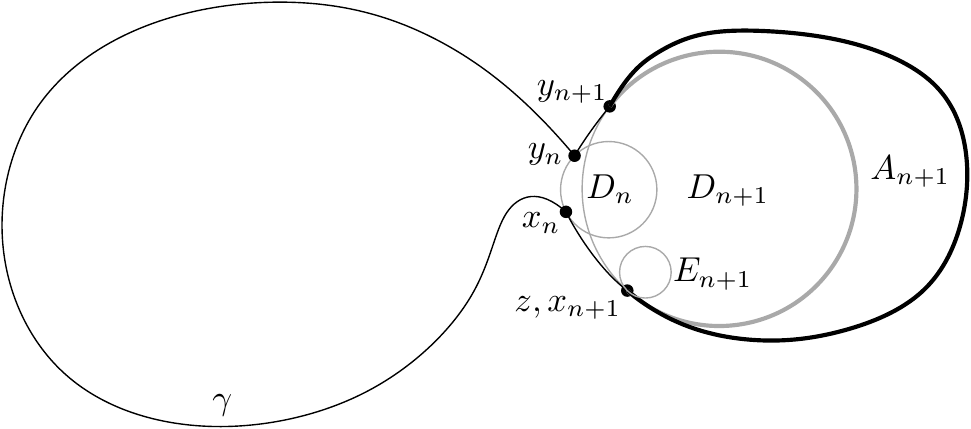}
\caption{The construction in the proof of Theorem~\ref{MAINTHM}, where $\gamma$ is the black Jordan curve.
The region $A_{n+1}$ is bounded by the fat Jordan curve.
The small disk $E_{n+1}$ is contained in $A_n$ (not excplicitly shown), but disjoint from $A_{n+1}$.}
\label{fig:thm}
\end{figure}

We have already constructed $x_0,y_0$ and $D_0$. We now describe the construction of $x_{n+1},y_{n+1}$, and $D_{n+1}$ given $x_n$, $y_n$, and $D_n$, and then argue that with this construction,~\ref{step1} and~\ref{step2} above are satisfied. Figure~\ref{fig:thm} illustrates the construction. 
First of all, $\gamma[x_n,y_n]$ winds positively around $D_n$ from $x_n$ to $y_n$ by Lemma~\ref{positivelemma}, so we may apply Lemma~\ref{diamLemma} and conclude that no open unit disk contains $\gamma[x_n,y_n]$. In particular, this applies to the open unit disk having the same center as $D_n$, and as the radius of $D_n$ is at most $1-\eta_2$, there exists a point $z\in \gamma(x_n,y_n)$ with $\dist(z,D_n)\geq \eta_2$.

Consider now the Jordan curve 
$$
\gamma_1:=\gamma[x_n,y_n]\cup \partial D_n[y_n,x_n]
$$
which, by Lemma~\ref{positivelemma}, contains $D_n$. We let $D_{n+1}$ be the open disk of maximal radius contained in $\Int \gamma_1$ and tangent to $U_{z}$ in $z$. 

By the same reasoning that we used to argue about $\partial D_0$ above, we must have that $\partial D_{n+1}$ meets $\gamma_1$ in at least two points.
None of these points can be in $\partial D_n(y_n,x_n)$ since this would imply that $D_{n+1}\subset D_n$ and hence that $z\in \overline{D_n}$, a contradiction. It follows that $|\gamma[x_n,y_n]\cap\partial D_{n+1}|\geq 2$. The set $\gamma\setminus \partial D_{n+1}$ is a collection of open intervals of $\gamma$, and since $|\gamma[x_n,y_n]\cap\partial D_{n+1}|\geq 2$, at least one of them, call it $\gamma(x_{n+1},y_{n+1})$, is contained in $\gamma(x_n,y_n)$. This completes the construction of $x_{n+1},y_{n+1}$, and $D_{n+1}$.


It remains to argue that with this construction, the conditions~\ref{step1} and~\ref{step2} are satisfied.

\begin{enumerate}[label=(\roman*)]
\item We make use of the following claim. 
\begin{claim}\label{claimthm1} 
We have that $\partial D_{n}[x_{n},y_{n}] \cap\partial D_{n+1}(x_{n+1},y_{n+1})=\emptyset$.
\end{claim}
\begin{proof}[Proof of Claim~\ref{claimthm1}]
Let the Jordan curve $\gamma_2$ be defined by
\begin{align*}
\gamma_2:=\gamma[x_{n+1},y_{n+1}]\cup \partial D_{n+1}[y_{n+1},x_{n+1}]. 
\end{align*}
By Lemma~\ref{positivelemma}, $D_{n+1}\subset\Int \gamma_2$, from which it follows that $\partial D_{n+1}(x_{n+1},y_{n+1})\subset \Int \gamma_2$.
Suppose for contradiction that $\partial D_{n}[x_{n},y_{n}]\cap\partial D_{n+1}(x_{n+1},y_{n+1})\neq \emptyset$. Since  $x_n,y_n\notin \Int \gamma_2$, $\partial D_{n}[x_{n},y_{n}]$ must then intersect $\gamma_2$ at least twice. 
But 
$$
\partial D_{n}[x_{n},y_{n}]\cap \gamma(x_{n+1},y_{n+1}) \subset \partial D_{n}[x_{n},y_{n}]\cap \gamma(x_n,y_n)=\emptyset,
$$
so in fact $\partial D_{n}[x_{n},y_{n}]$ must intersect $\partial D_{n+1}[y_{n+1},x_{n+1}]$ at least twice. 
It follows that $\partial D_{n}$ intersects $\partial D_{n+1}$ at least \emph{thrice}, which is a contradiction as $D_n\neq D_{n+1}$. 
We conclude that $\partial D_{n}[x_{n},y_{n}] \cap \partial D_{n+1}(x_{n+1},y_{n+1})=\emptyset$, as desired. 
\end{proof}
The arc $\partial D_n[x_n,y_n]$ separates $\Int \gamma_1$ into two regions, namely $D_n$ and $A_n$. The claim thus gives that either $\partial D_{n+1}(x_{n+1},y_{n+1})\subset D_n$ or $\partial D_{n+1}(x_{n+1},y_{n+1})\subset A_n$.
Now observe that $x_{n+1}\in \gamma(x_n,y_n)$ or $y_{n+1}\in \gamma(x_n,y_n)$: Indeed, $z\in \gamma(x_n,y_n)\cap \overline{D_{n+1}}$ but $\gamma(x_{n+1},y_{n+1})$ contains no point of $\overline{D_{n+1}}$. 
In particular either $x_{n+1} \notin \overline{D_n}$ or $y_{n+1} \notin \overline{D_n}$ and it is therefore the case that $\partial D_{n+1}(x_{n+1},y_{n+1})\subset A_n$. Since now $\partial A_{n+1}\subset \overline{A_n}$, we get that $A_{n+1}\subset A_n$.

\item We define $E_{n+1}$ to be the disk of radius $\eta$, tangent to $U_{z}$ in $z$, and contained in $U_{z}$.
The radius of $E_{n+1}$ is at most $\eta_2/2$, and since $z\in \partial E_{n+1}$ has distance at least $\eta_2$ to $D_n$, it follows that $E_{n+1}\subset \Int \gamma_1 \setminus \overline{D_n}=A_n$.
Moreover, $E_{n+1}\cap A_{n+1}\subset D_{n+1}\cap A_{n+1}=\emptyset$, and we conclude that $E_{n+1}\subset A_n\setminus A_{n+1}$, as desired. 
\end{enumerate}
Having argued that the conditions~\ref{step1} and~\ref{step2} are satisfied, the proof is complete.
\end{proof}

\section{Open problems}
We mention here two open problems that we find interesting.

\subsection{Are curves of bounded convex curvature rectifiable?}
As mentioned in the introduction, some earlier proofs of the Pestov--Ionin theorem have used that the length of $\gamma$ is finite.
In contrast, our proof relies on $\Int \gamma$ having finite area which is an immediate property of Jordan domains.
It is, however, easy to verify that \emph{if} curves of bounded convex curvature are rectifiable, i.e., has finite length, then the proof given by Pestov and Ionin~\cite{pestov1959largest} would carry through almost unchanged.
We believe this to actually be the case.
Is there a (simple) proof that curves of bounded convex curvature are rectifiable?

\subsection{What is the picture in higher dimensions?}
The Jordan-Brouwer separation theorem states that if $\gamma$ is an $n$-dimensional topological sphere in $\RR^{n+1}$, i.e., is obtained as the image of an injective continuous map $S^n\longrightarrow \RR^{n+1}$, then the complement of $\gamma$ in $\RR^{n+1}$ consists of exactly two connected components, one being bounded (the interior) and one being unbounded (the exterior). 

It is easy to generalize our notion of bounded convex curvature to this setting.
We say that  $\gamma$ has \emph{bounded convex curvature} if for every point $x$ on $\gamma$, there is an open $(n+1)$-dimensional unit ball $U_x$ and $\ee_x>0$ such that
\begin{align}\label{bccCondGen}
x\in\partial U_x\quad\text{and}\quad \ball{x}{\ee_x}\cap U_x\subset\Int\gamma.
\end{align}
The natural question is: If $\gamma$ has bounded convex curvature, does $\Int \gamma$ contain an open $(n+1)$-dimensional unit ball? 
This turns out to be false.
Indeed, Lagunov and Fet~\cite{Lagunovsphere,Lagunovsphere2} studied connected $n$-dimensional $C^2$-hypersurfaces in $\RR^{n+1}$ having all principal curvatures $|\kappa_i|\leq 1$.
They showed, for instance, that topological $n$-spheres with these properties  all contain an $(n+1)$-ball in their interior of radius at least $r_0=\sqrt{3/2}-1 \cong 0.2246$, and that this is sharp when $n=2$.
Other relevant work was made by Lagunov~\cite{Lagunovsurf3,Lagunovsurf,Lagunovsurf2}, who showed that all compact, connected, $C^2$, $n$-dimensional hypersurfaces embedded in $\RR^{n+1}$, for which all principal curvatures $\kappa_i$ satisfy $|\kappa_i|\leq 1$, contain a ball of radius $r_1=2/\sqrt{3}-1\cong 0.155$ and that this is sharp. 

As our class of hypersurfaces of bounded convex curvature is less restricted (there is no assumption on differentiability and we make no requirement that the concave curvature be bounded) it is natural to ask whether it still holds that 
topological $n$-spheres of bounded convex curvature contain a ball of radius $r_0$ (or $r_1$ in the case of general compact, connected, $n$-dimensional hypersurfaces embedded in $\RR^{n+1}$) in their interior. Even for $n=2$ we find this an interesting question.

\subsection*{Acknowledgments}
We thank Anders Thorup for his very careful reading of the manuscript and numerous suggestions for improving the presentation, in particular by pointing out steps in our proofs that seemed intuitively clear, but in fact required detailed arguments.

A big obstacle in our work has been that many of the relevant papers are written in Russian. We thank Richard Bishop for providing us copies of his English translations of~\cite{Lagunovsphere} and~\cite{Lagunovsphere2}.
We furthermore wish to thank the teams behind \href{www.i2ocr.com} and \href{translate.google.com}, the first of which we used to convert the cyrillic script in~\cite{pestov1959largest} to machine-encoded cyrillic text and the second of which to translate the resulting text to English, together making it possible for us to understand the proof given by Pestov and Ionin.






\newpage

\appendix

\section{Appendix}~\label{sec:appen}
In this appendix we will discuss orientations of Jordan curves and eventually provide a proof of Lemma~\ref{leftrightlemma}.
For this purpose it will be necessary to view curves as continuous maps $\varphi:I \longrightarrow \RR^2$ where $I=[t_1,t_2]\subset\RR$ is a closed and bounded interval. With this notation $\varphi$ is closed if $\varphi(t_1)=\varphi(t_2)$, and simple if $\varphi$ is injective, where in the closed case we allow $\varphi(t_1)=\varphi(t_2)$. A Jordan curve $\gamma$ is the image of a simple closed curve.

The starting point will be the classic theorem by Jordan.
\begin{theorem}[Jordan curve theorem]
Let $\gamma$ be a Jordan curve. Then the complement $\RR^2\setminus \gamma$ consists of two connected components. Moreover, $\gamma$ is the boundary of each of these components. 
\end{theorem}

We now recall some simple facts concerning argument variation.
For a given point $p\in \RR^2$ and $x\neq p$ an \emph{argument} for $x$ with respect to $p$ is an argument for the vector $x-p$, that is, an angle $\theta\in\RR$ such that $x-p=(r \cos\theta,r \sin\theta)$ for some $r>0$.
If $I=[t_1,t_2]$, $\varphi:I\longrightarrow \RR^2$ is a curve, and $p\notin  \varphi(I)$, then a \emph{continuous argument function} for $\varphi$ with respect to $p$ is a continuous map $\theta:I \longrightarrow \RR^2$ such that $\theta(t)$ is an argument for $\varphi(t)$ for all $t\in I$.
The \emph{argument variation} around $p$ is defined as $\Arg_p \varphi=\theta(t_2)-\theta(t_1)$ and this does not depend on the choice of $\theta$, nor is it changed if we use an orientation preserving reparametrization of $\varphi$. Importantly, the function $p \longmapsto \Arg_p \varphi$ is continuous on $\RR^2 \setminus \varphi(I)$.

From the above it follows that for a Jordan curve $\gamma$, the argument variation around any $p\notin\gamma$ is a multiple of $2\pi$,  constant on each of the two connected components of the complement of $\gamma$, and $0$ on the unbounded component.
In fact, the argument variation is $\pm 2 \pi$ when $p\in \Int \gamma$, as we will see shortly.
We say that a parametrization $\varphi$ of $\gamma$ is \emph{positively oriented} if the argument variation of $\varphi$ around any $p\in\Int\gamma$ is $2\pi$.
Otherwise we say that $\varphi$ is \emph{negatively oriented}.
If $a$ and $b$ are distinct points on $\gamma$, we write $\gamma[a,b]$ for the interval of $\gamma$ obtained by traversing $\gamma$ from $a$ to $b$ along the positive orientation.
Slightly abusing notation we will sometimes write $\gamma[a,b]$ for a parametrization of this interval traversed from $a$ to $b$.
We also define $\gamma(a,b)=\gamma[a,b]\setminus\{a,b\}$.

Finally, if $\varphi_1:[s_1,s_2]\longrightarrow \RR^2$ and $\varphi_2:[t_1,t_2]\longrightarrow \RR^2$ are curves satisfying $\varphi_1(s_2)=\varphi_2(t_1)$, we let $\varphi_1 +\varphi_2$ be the continuous curve obtained by first traversing $\varphi_1$ and then $\varphi_2$.
Also, if $\varphi$ is a curve, we write $-\varphi$ for the curve obtained by traversing $\varphi$ in the opposite direction.
We finally write $\varphi_1-\varphi_2=\varphi_1+(-\varphi_2)$ when the addition is well-defined.
We now restate Lemma~\ref{leftrightlemma} in a more general form.

\begin{figure}
\centering
\includegraphics{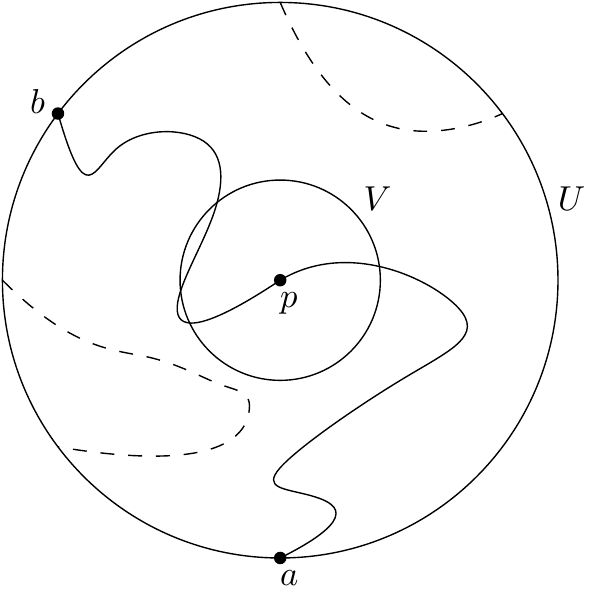}
\caption{The setting of Lemma~\ref{leftrightlemma2}. The dashed lines represent $\gamma$ potentially reentering $U$.}
\label{fig:leftright}
\end{figure}

\begin{lemma}\label{leftrightlemma2}
Let $p$ be a point on a Jordan curve $\gamma$, and let $U$ be an open disk with center $p$, sufficiently small so that $\gamma$ is not contained in $U$.
The intersection of $U$ and $\gamma$ is a collection of open intervals of $\gamma$ of which one, say $\gamma(a,b)$, contains $p$.
Consider the two Jordan curves
$$
\alpha^+=\gamma[a,b]+ \partial U[b,a] \quad \text{and} \quad \alpha^-=\gamma[a,b]- \partial U[a,b].
$$
Then $U$ is the disjoint union
$$
U=\gamma(a,b) \cup \Int \alpha^+ \cup \Int \alpha^-.
$$
Moreover, $\Int \gamma$ and $\Int \alpha^+$ coincide near $p$, that is, there exists a small disk $V\subset U$ centered at $p$ such that $\Int \gamma \cap V=\Int \alpha^+ \cap V$. (See Figure~\ref{fig:leftright}.)
\end{lemma}

\paragraph{Remark.}
We prove the lemma by first proving the statement about the decomposition of $U$.
Strictly speaking, $\gamma(a,b)$ is not defined at this point and $a$ and $b$ may a priori be chosen in two different ways.
Likewise, $\partial U[b,a]$ and $\partial U[a,b]$ are not defined, but it is obvious what it means to traverse a circle in the positive and negative direction.
The statement about the decomposition of $U$ is correct regardless of how $a$ and $b$ are chosen, so the ambiguity does not matter for that part of the lemma.
The second part of the proof starts by showing that the argument variation around any $q\in \Int \gamma$ is $\pm 2 \pi$, which by the introductory comments lets us define intervals such as $\gamma[a,b]$ unambiguously.
This is important for the statement that $\Int\gamma$ and $\Int \alpha^+$ coincide near $p$, which is finally proven.

\begin{proof}[Proof of Lemma~\ref{leftrightlemma2}]
Define $\delta= \partial U$. For any point $x\in U$ let $\theta_x\in (0,2\pi)$ be the angle from $a-x$ to $b-x$ in the positive direction. Clearly, $\theta_x=\Arg_x \delta[a,b]$ and $\Arg_x \delta[b,a]=2 \pi-\theta_x$. Hence, if $x$ is in $U$ and not on $\gamma(a,b)$ we have the two equations
\begin{align}
\Arg_x \alpha^+&=\Arg_x \gamma[a,b]+2\pi-\theta_x, \label{eq1}\\
 \Arg_x \alpha^{-}&=\Arg_x \gamma[a,b]-\theta_x. \label{eq2}
\end{align}
If $x\in \Ext \alpha^+$ then the left side of~\eqref{eq1} vanishes, and if $x \in \Ext \alpha^-$ then the left side of~\eqref{eq2} vanishes. Consequently, by~\eqref{eq1} and~\eqref{eq2},
\begin{align}
\Arg_x \gamma[a,b]&=\theta_x-2 \pi &\text{for } x\in \Ext \alpha^+ \cap U, \label{eq3}\\
\Arg_x \gamma[a,b]&=\theta_x &\text{for } x\in \Ext \alpha^- \cap U. \label{eq4}
\end{align}
In turn, when the latter two equations are inserted into~\eqref{eq1} we obtain the equations
\begin{align}
\Arg_x \alpha^+&=2 \pi &\text{for } x\in \Ext \alpha^- \cap U,\label{eq5} \\
\Arg_x \alpha^{-}&=-2 \pi &\text{for } x\in \Ext \alpha^+ \cap U. \label{eq6}
\end{align}
We now observe that $\Ext \alpha^- \cap U$ is nonempty, as follows.
Let $q$ be a point on $\delta(b,a)$ and $W$ an open disk centered at $q$ so small that $W\cap \alpha^{-}=\emptyset$.
Since a part of $W$ is outside $U$ (and thus in $\Ext\alpha^-$), it follows that $W\subset \Ext \alpha^-$. 
Hence $W\cap U\subset \Ext \alpha^{-} \cap U$ and since $W\cap U$ is nonempty, the claim follows.
Let $y\in\Ext \alpha^- \cap U$.

We next prove that $\Ext \alpha^- \cap U = \Int \alpha^+$.
Consider a point $x\in \Ext \alpha^- \cap U$.
It follows from~\eqref{eq5} that $\Arg_x \alpha^+=2 \pi$ and hence that $x\in\Int \alpha^+$. 
On the other hand, consider a point $x\in \Int \alpha^+$.
Since $\Arg_y \alpha^+=2\pi$ and $z\longmapsto \Arg_z \alpha^+$ is constant on $\Int \alpha^+$, it follows that $\Arg_x \alpha^+=2 \pi$.
We now get from~\eqref{eq1} that $\Arg_x \gamma[a,b]=\theta_x$, and then by~\eqref{eq2} that $\Arg_x \alpha^-=0$, that is; $x \in \Ext \alpha^-$.
We conclude that $\Ext \alpha^- \cap U = \Int \alpha^+$, as claimed.
We can in a similar way show that $\Ext \alpha^+ \cap U = \Int \alpha^-$.
The assertion in the lemma concerning the decomposition of $U$ then follows.

To prove the final assertion, choose an open disk $W \subset U$ centered at $p$ such that $W\cap \gamma[b,a]=\emptyset$. 
For any point $x$ not on $\gamma[b,a]$, let $v_x:= \Arg_x \gamma[b,a]$.
When $x\notin\gamma$, we have $\Arg_x \gamma =v_x+\Arg_x \gamma [a,b]$.
Hence, by~\eqref{eq3} and~\eqref{eq4}
\begin{align}
\Arg_x \gamma&=v_x+\theta_x-2\pi &\text{for } x\in \Int \alpha^- \cap W&, \label{eq7} \\
\Arg_x \gamma&=v_x+\theta_x &\text{for } x\in \Int \alpha^+ \cap W& \label{eq8}.
\end{align}
Note that $x\longmapsto v_x$ and $x\longmapsto \theta_x$ are defined and continuous on all of $W$.
Choose the open disk $V\subset W$ centered at $p$ such that $|(v_x+\theta_x)-(v_p+\theta_p)|<\pi$ for $x\in V$, so that every value of~\eqref{eq7} is strictly smaller than every value of~\eqref{eq8}.
The argument variation $\Arg_x \gamma$ is constant on the two connected components of $\RR^2 \setminus \gamma$, so it can take two possible values.
Hence, the values of~\eqref{eq7} and~\eqref{eq8} must be constant, and with $A:=v_p+\theta_p$ it follows that $v_x+\theta_x=A$ for $x\in V$ and that
\begin{align*}
\Arg_x \gamma&=A-2\pi &\text{for } x\in \Int \alpha^- \cap V&, \\
\Arg_x \gamma&=A &\text{for } x\in \Int \alpha^+ \cap V&.
\end{align*}
The second value is $2\pi$ larger than the first, and, a priori, one of the values is $0$.
Hence, either the values are $-2\pi$ and $0$ or they are $0$ and $2\pi$.
We now say that $\gamma$ is \emph{positively oriented} if the values are $0$ and $2\pi$ and \emph{negatively oriented} in the other case.
Now assume that $a$ and $b$ are chosen such that $\gamma[a,b]$ is the interval obtained by traversing $\gamma$ from $a$ to $b$ along the positive orientation. Then it follows that $\Arg_x \gamma=2\pi$ for all $x\in \Int \gamma$ and that $\Int \gamma \cap V=\Int \alpha^+\cap V$, as asserted.

\end{proof}

\end{document}